\date{May 23, 2014}
\documentclass[12pt]{amsart}
\usepackage{latexsym,amsmath,amsfonts,amscd,amssymb}
\usepackage{graphics}
\newtheorem{theorem}{Theorem} %[subsection]
\newtheorem{lemma}[theorem]{Lemma}

\newtheorem{corollary}[theorem]{Corollary}

\theoremstyle{remark}

\newcommand{\cM}{{\mathcal M}}

\newcommand{\cS}{{\mathcal S}}

\newcommand{\cW}{{\mathcal W}}

\newcommand{\RR}{{\mathbb R}}

\newcommand{\ZZ}{{\mathbb Z}}

\title{A simple dynamical model leading to Pareto wealth distribution and stability}


\author[R. P\'{e}rez Marco]{Ricardo P\'{e}rez Marco}
\address{CNRS, LAGA UMR 7539, Universit\'e Paris XIII,
99, Avenue J.-B. Cl\'ement, 93430-Villetaneuse, France}
\normalsize\email{ricardo.perez.marco@gmail.com}

\begin{document}
\begin{abstract}
We propose a simple dynamical model of wealth evolution. The invariant
distributions are of Pareto type and are dynamically stable as conjectured by Pareto.
\end{abstract}

\maketitle
%\noindent \emph{We dedicate this article to }

\section{Introduction.}

At the end of the XIXth century, in his studies of wealth and 
income \footnote{Wealth and income are proxies of each other in first approximation for our purposes.} 
distribution on different countries, Vilfredo 
Pareto (\cite{P1}, \cite{P2}) discovered 
the universal power law that governs the upper tail of wealth distribution. It is well known that this is not a 
good model for the lower part of the curve that is more dependent on specific sociological factors and of log-normal type 
(see the discusion in \cite{Ma}).  
The exponent in the power decay is country dependent 
and is an indicator of equitative wealth (re)distribution. A larger exponent indicates a more equitative wealth distribution. 
Pareto's universal assymptotic behaviour appears in distributions
from various other contexts, and, as we show, is typical from competitive system where the reward is proportional to the accumulated 
wealth. The purpose of this article is to provide a simple explanation to Pareto's empirical observation. We propose 
a natural dynamical model of evolution of wealth where Pareto distributions emerge as invariant dynamically stable 
\footnote{``Dynamically stable distribution'' in the Dynamical System sense not in the probabilistic sense.}
distributions of this Dynamical System. The stability of the wealth distribution, which is different from the universality property, 
was conjectured by Pareto, whose intuition apparently comes from his empirical observations. We can read in \cite{P2}, chap. VII, point 31, p.393:

\medskip

\textit{Si, par exemple, on enlevait tout leur revenu aux citoyens les plus riches, en supprimant la queue de la figure des revenus, celle-ci
ne conserverait pas cette forme, mais t\^ot ou tard elle se r\'etablirait suivant une forme semblable \`a la premi\`ere.}\footnote{\textit{``If, for instance, we confiscate all income to the richests citizens, thus erasing the tail of income distribution, this shape will not persist and sooner or later it will evolve to a similar shape of the original.''}}

\medskip

There are other classical models and studies of Pareto empirical observation and power laws (like Zipf's law). For the record 
we cite a few classical ones: H. Simon \cite {Si}, D.G. Champernowne \cite{Ch}, B. Mandelbrot \cite{Ma},etc Simon
model \cite{Si} for Zipf's law is a ``genesis model'' of the distribution, i.e. it is a model for its creation.  
Champernowne \cite{Ch} proposed a general multiplicative stochastic model, and B. Mandelbrot \cite{Ma} explained Pareto law 
by the universal limit character of Pareto-L\'evy probabilistically stable distributions.

\section{The dynamical model.}

In this first section, we propose and study a dynamical model of wealth evolution which is a simple first 
approximation. 

\subsection{Setup.}

Let $f(x)$ be the wealth distribution, i.e. $df= f(x) \  dx$ is the number of individuals with wealth 
in the infinitesimal interval $[x, x+dx[$. The 
distribution function 
$f : \RR_+ \to \RR_+$ is continuous, positive and decreasing and $\lim_{x\to +\infty} f(x) =0$. A distribution is of Pareto type 
if it presents a power law decay $x^{-\alpha}$ at $+\infty$, that is
$$
\lim_{x\to +\infty} -\frac{\log f(x)}{\log x} =\alpha >0 \ .
$$
The exponent $\alpha >0$ is \textit{Pareto exponent}. A distribution of the form $f(x) = C. x^{-\alpha}$ is called a Pareto distribution.
Smaller values of $\alpha$ indicate larger inequalities in wealth distribution. Notice that $\alpha >1$ is necessary for the distribution to 
be summable at $+\infty$, i.e. finite wealth at infinite (finitness near $0$ is not significant since the model aims to explain the tail behaviour at $+\infty$).

\subsection{Wealth dynamics.}

We focuss on the evolution of individual wealth. We assume that the evolution is based on two main factors:  Finantial decisions, that we model 
as a betting game, and by public redistribution of wealth, that absorbs part of the individual wealth into public wealth.

For the first factor we model the finantial decisions of each individual   
by a sequence of bets. Each financial decision turns out to be 
a bet, waging a proportion of his wealth. As a first approximation, we assume that the 
probability of success is the same for all agents and bets $0<p<1$ (this is the average probability). At each round, 
each agent risks the same percentage of his wealth, a fraction $\gamma >0$ (that is also an average). If 
he wins, his wealth is multiplied by the factor $1+\gamma$ and if he looses his wealth is divided by $1+\gamma$.

Only considering this first factor, one round evolution the distribution transforms into the new distribution
$$
\cW (f) (x) = \frac{p}{1+\gamma} \ f( x/(1+\gamma)) + (1-p)(1+\gamma ) \ f( (1+\gamma) x) \ .
$$
The operator $\cW$ is ``wealth preserving''. In terms of $L^1$-norm we have
$$
||\cW (f)||_{L^1} = ||f||_{L^1} \ .
$$
The agents will only risk their capital if there is a positive expectation of gain, thus we should assume that $p>1/2$.

There are other mechanisms that affect wealth evolution that we should consider, 
as for example inheritances that divide wealth, taxes, etc. Note that public wealth drains individual wealth by the 
fiscal mechanism.
Thus it is natural to consider a broader class of operators $\cW$ with a dissipative parameter $\kappa \geq 1$, the \textit{dissipative coefficient},
$$
\cW_\kappa (f) (x) = \frac{1}{\kappa} \cW (f) (x)=\frac{p}{\kappa (1+\gamma)}f \left ( x/(1+\gamma) \right ) + \frac{(1-p)(1+\gamma)}{\kappa} f( (1+\gamma)x) \ ,
$$
so that for $\kappa=1$ the operator is wealth preserving. We name the model for $\kappa =1$ the ``wealth preserving model''.

\subsection{Invariant distributions.}

Distributions invariant by the evolution operator $\cW_\kappa$  must satisfy 
the fixed point functional equation $\cW_\kappa (f)=f$, that is,
\begin{equation} \label{functional_eq0}
f(x)= \frac{p}{\kappa (1+\gamma)}f \left ( x/(1+\gamma) \right ) + \frac{(1-p)(1+\gamma)}{\kappa} f( (1+\gamma)x)   \ .
\end{equation}
We solve this equation in the next section. 

\subsection{Solution of the functional equation.}

Considering the change of variables $F(x)=f(e^x)$, equation (\ref{functional_eq0}) 
becomes a functional equation for $F:\RR \to \RR$
\begin{equation} \label{functional_eq}
a \ F(x+\lambda) - F(x) +b \ F(x-\lambda) =0  \ ,
\end{equation}
where $\lambda = \log (1+\gamma ) >0$, $a=(1-p)(1+\gamma )/\kappa > 0$ and $b=p/\kappa /(1+\gamma)  > 0$.

We have a general theory of these type of functional equations. L. Schwartz (see \cite{S}, and also \cite{M}, \cite{K})
studied more general ``mean periodic'' smooth functions $F$ 
that satisfy a functional equation of the form
$$
\omega \star F = 0 \ ,
$$
where $\omega$ is a compactly supported distribution. In our case, $\omega = a \ \delta_\lambda -\delta_0 + b \ \delta_{-\lambda}$\footnote{The way to study these equations is by Fourier transforming it (\`a la Carleman \cite{C} using hyperfunctions
in order to work in sufficient generality).
One of the general results by L. Schwartz (see \cite{S} Theorem 10 p.894) is the ``spectral synthesis'' of solutions: Smooth solutions are uniform 
limits on compact set of $\RR$ of linear combinations of exponential solutions $(e^{\rho x})_\rho$. Also these exponential 
solutions are not limits of linear combinations of the others, 
thus the expansion is unique.}.
\medskip

In our simplified model we don't need the general theory and the equation can be solved by elementary 
means. 
First, the exponential solutions are easy to calculate. 
A function $F(x)=e^{\rho x}$ is a solution if $e^{\rho \lambda}$ satisfies
the following second degree equation:
\begin{equation}\label{2nd_degree}
a \left (e^{\rho \lambda}\right )^2 - \left (e^{\rho \lambda}\right ) +b =0 \ .
\end{equation}
Observe that the discriminant $\Delta = 1-4ab$ is positive since we have
$$
ab=\frac{p(p-1)}{\kappa^2}< \frac{1}{4\kappa^2}  \ ,
$$
thus $\Delta > 1-\frac{1}{\kappa^2} > 0$ because $\kappa \geq 1$. 

Thus we have two distinct solutions:

$$
e^{\rho \lambda} = \frac{1}{2a} \pm \frac{1}{2a} \sqrt{1-4ab} \ . 
$$

Since $a>0$ and the polynomial $P(x)=ax^2-x+b$ satisfies $P(0) >0$ 
and $P(1)<0$, we have two real root $x_1$ and $x_2$ with $0<x_1 < 1 < x_2$.
Therefore, we have two families of solutions for $\rho$ in two vertical lines in the complex domain, for $k\in \ZZ$, $j=0,1$, 
$$
\rho_{j,k} = \lambda^{-1}\log x_j  + 2 \pi i k \lambda^{-1} \ .
$$
Note that $\Re \rho_{0,k} < 0 < \Re \rho_{1,k}$. Let $\rho_0= \rho_{0,0} <0$ and $\rho_1=\rho_{1,0}>0$.
Observe that the particular solution $F(x)= C.e^{\rho_{0} x}$  
leads to the solution $f(x)=F(\log x)=C. x^{\rho_0}$ which is exactly Pareto distribution with Pareto 
exponent $\alpha = -\rho_0$.  

We can now solve the functional equation completely\footnote{We have a strong form of Schwartz spectral theorem.}:

\begin{theorem}
The general solution of the functional equation (\ref{functional_eq}),
\begin{equation} 
a \ F(x+\lambda) - F(x) +b \ F(x-\lambda) =0  \ ,
\end{equation}
(with $a, b, \lambda$ as above) is 
$$
F(x) = e^{\rho_0 x} L_0(x/\lambda) +e^{\rho_1 x} L_1(x/\lambda)
$$
where $L_0$ and $L_1$ are $\ZZ$-periodic functions.
\end{theorem}

In order to solve the functional equation (\ref{functional_eq}), 
we consider $H(x) =F(x+\lambda)-e^{\rho_0 \lambda} F(x)$. Substracting 
(\ref{functional_eq}) from (\ref{2nd_degree})  multiplied by $e^{-\rho_0 \lambda} F(x)$ we get
$$
aH(x)-be^{-\rho_0 \lambda} H(x-\lambda) =0 \ ,
$$
or
$$
H(x)=\left (\frac{b}{a} \  e^{-\rho_0 \lambda} \right ) H(x-\lambda) \ .
$$
Considering 
$$
\hat H(x) =\left (\frac{b}{a} \ e^{-\rho_0 \lambda} \right )^{-x/\lambda} H(x) \ ,
$$
we have that $\hat H(x)=\hat H(x-\lambda)$, i.e. there is a $\ZZ$-periodic function $L$ such that
$$
H(x)= \left (\frac{b}{a} \ e^{-\rho_0 \lambda} \right )^{x/\lambda} L(x/\lambda) \ .
$$
Therefore we have
$$
F(x+\lambda)-e^{\rho_0 \lambda} F(x)= \left (\frac{b}{a} \ e^{-\rho_0 \lambda} \right )^{x/\lambda} L(x/\lambda) \ .
$$
Now, put
$$
\hat F(x) = e^{-\rho_0 x} F(x) \ .
$$
Then we need to solve
$$
\hat F(x+\lambda) - \hat F(x) = e^{-\rho_0 \lambda} \ \left (\frac{b}{a} \right )^{x/\lambda} \  e^{-2\rho_0 x} L(x/\lambda) \ ,
$$
if we write $G(x)=e^{\rho_0 \lambda} \hat F(x)$ and $c=-2\rho_0+\lambda^{-1} \log(b/a)$,
$$
G(x+\lambda) - G(x) =e^{cx} L(x/\lambda) \ .
$$
We use the following lemma:

\begin{lemma}
For $c\in \RR$, $\lambda >0$, and $L$ a $\ZZ$-periodic function, the solutions of the functional equation
\begin{equation}
 G(x+\lambda)-G(x) = e^{cx} L(x/\lambda)\ ,
\end{equation}
are of the form
$$
G(x)= G_0(x)+M(x/\lambda) \ ,
$$
where $M$ is a $\ZZ$-periodic function, and for $c\not= 0$,
$$
G_0(x)= \frac{e^{cx}}{e^{c\lambda}-1} \  L(x/\lambda) \ ,
$$
and for $c=0$ 
$$
G_0(x)= \lambda^{-1} x \ L(x/\lambda) \ .
$$
\end{lemma}

\begin{proof}
Obviously in both cases $G_0$ is a particular solution. Then the functional equation 
is equivalent to $M(x+1)-M(x)=0$, where $M(x)=G(\lambda x)-G_0(\lambda x)$,  i.e. $M$ is $\ZZ$-periodic.
\end{proof}

So, in the non-degenerate case ($c\not= 0$), absorbing the multiplicative constants into $L$ and $M$, 
the general solutions of (\ref{functional_eq}) are of the form
$$
F(x) = e^{(-\rho_0 +\lambda^{-1} \log (b/a) ) x} L(x/\lambda) +e^{\rho_0 x} M(x/\lambda) \ .
$$
And coming back to the second degree equation (\ref{2nd_degree}) we have
$$
e^{-\rho_0 \lambda} \ \frac{b}{a}=e^{\rho_1 \lambda} \ ,
$$
so
$$
F(x) = e^{\rho_1 x} L(x/\lambda) +e^{\rho_0 x} M(x/\lambda) \ .
$$
Indeed the degenerate case never happens:

\begin{lemma}
We have $c\not=0$.
\end{lemma}

\begin{proof}
If $c=0$ then $e^{2\rho_0 \lambda}=b/a=e^{\rho_0 \lambda}.e^{\rho_1 \lambda}$ and $e^{\rho_0 \lambda}=e^{\rho_1 \lambda}$, the 
root of the second degree equation would be double and the discriminant would be $\Delta =0$ but we have seen that $\Delta >0$.
\end{proof}

If we request that $F>0$ and $F(x)\to 0$ for $x\to +\infty$ (the only sound solutions) then $L=0$ and $M>0$,
$$
F(x) = e^{\rho_0 x} M(x/\lambda) \ .
$$
Finally we have
$$
f(x)=x^{\rho_0} M(\lambda^{-1}\log x ) \ .
$$
If we look for continuous solutions, then $M$ must be continuous and bounded since it is $\ZZ$-periodic,
thus $f$ satisfies Pareto assymptotics
$$
\lim_{x\to +\infty} -\frac{\log f(x)}{\log x} = - \rho_0 = \alpha >0 \ .
$$

\subsection{Pareto exponent.}
It is interesting that we can compute an explicit expression of the Pareto exponent in terms of the parameters $\kappa$, $\gamma$ and $p$,
\begin{corollary}
The Pareto exponent is given by
$$
\alpha = -\rho_0 = -\lambda^{-1} \log \left (\frac{1-\sqrt{1-4ab}}{2 a}\right )
$$
 or
$$
\alpha = 1- \frac{\log \left (\frac{\kappa-\sqrt{\kappa^2-4p(1-p)}}{2 (1-p)}\right )}{\log(1+\gamma) } \ .
$$
\end{corollary}

It is interesting to note that the Pareto exponent $\alpha$ decreases when $\gamma$ increases. This means that a more risky 
finantial behaviour, or more active economy, favours unequal distribution. Fortunes are created and lost more often. Ruin 
is more common. Indeed we know by the Kelly criterion \cite{Ke} that ruin is almost sure in the long run if $\gamma$ is larger than a certain threshold. 
With a slightly modified model we can explain Pareto's theory of ``Circulation of Elites''. Indeed this circulation occurs at all level 
of social status when the agents are not enough conservative to satisfy Kelly criterion. We will discuss these questions in a companion article \cite{PM}.

The Pareto exponent also increases with $\kappa$ since
$$
\frac{d\alpha}{d\kappa} =\frac{1}{\log(1+\gamma)} \frac{\kappa-\sqrt{\kappa^2-4p(1-p)}}{\sqrt{\kappa^2-4p(1-p)} \left (\kappa-\sqrt{\kappa^2-4p(1-p)}\right )} \ ,
$$
is positive. This is natural since a larger $\kappa$ means a larger demographic and fiscal pressure and thus we expect 
a better redistribution of wealth and a larger Pareto exponent. 

\subsection{A remarkable solution in the wealth preserving model.}

A Pareto exponent $\alpha >1$ is necessary for summability of the tail of the distribution and is always observed in experimental studies. 
It is remarkable that in the wealth preserving model with the critical value of the dissipative coefficient $\kappa =1$, 
the Pareto exponent is exactly $\alpha =1$.

\begin{theorem}
 In the wealth preserving model, $\kappa = 1$, the Pareto exponent is exactly equal to $\alpha =1$.
\end{theorem}

\begin{proof}
 For $\kappa=1$ we have
$$
\kappa^2-4p(1-p)=(2p-1)^2 \ .
$$
Therefore
$$
\kappa-\sqrt{\kappa^2-4p(1-p)} =2(1-p) \ .
$$
And the formula in the previous section gives $\alpha =1$. 
\end{proof}

This result is natural and to be expected: For $\kappa < 1$ the wealth in increasing without limit and the invariant distributions 
could not be summable at $+\infty$, and for $\kappa > 1$ we have finite wealth at $+\infty$. From the form of the invariant solutions, we have:

\begin{theorem}
 For an invariant solution, the following conditions are equivalent:
\begin{enumerate}
 \item The tail wealth is summable, $W(f, x_0) < +\infty$ .
 \item The Pareto exponent $\alpha$ is larger than $1$, $\alpha >1$.
 \item The model is wealth dissipative, that is $\kappa >\kappa_0$ .
\end{enumerate}
\end{theorem}

It has been observed that the Pareto exponent of the wealthiest fraction of the population has a Pareto exponent which is much closer to $1$ 
than expected (or to the rest of the medium class, whatever this means). So for this class of the population the dissipative coefficient 
is closer to the critical one $\kappa_0$, this means that the wealthiest part of the population is able to avoid the mechanisms of 
fiscal redistribution of wealth.

\subsection{Stability of invariant solutions.}
We now study the Pareto problem of stability of the Pareto distribution.

Since $\kappa >1$, we can observe that for the $L^1$-norm the operator $\cW_\kappa$ is contracting:
\begin{lemma}
Let $f, g : \RR_+^* \to \RR_+$ be measurable functions , with $f-g \in L^1(\RR_+^* )$, then 
$$
|| \cW_\kappa (f) -\cW_\kappa (g)||_{L^1} \leq \kappa^{-1} ||f-g||_{L^1} \ .
$$
\end{lemma}
\begin{proof}
We have
 \begin{align*}
 \left | \cW_\kappa (f) (x) -\cW_\kappa (g) (x) \right | \leq  & 
 \frac{p}{\kappa (1+\gamma)} \left |f(x/(1+\gamma)) -g(x/(1+\gamma))\right | \\ 
&+\frac{(1-p)(1+\gamma)}{\kappa}\left |f(x(1+\gamma)) -g(x(1+\gamma))\right | 
 \end{align*}
and the result follows integrating over $\RR_+^*$.
\end{proof}

Obviously this lemma is only interesting when $||f-g||_{L^1}$ is finite. 
For each invariant solution $f_0$ it is natural to consider the space of measurable bounded perturbations of $f_0$ for the $L^1$-norm, $\cM( \RR_+^*, \RR )$ denotes
the space of Borel measurable functions,
$$
\cS_{f_0} = \{ g \in \cM( \RR_+^*, \RR )  ; ||g-f_0||_{L^1}<+\infty \} \ .
$$
Then the fixed point $f_0$ is a global attractor in $\cS_{f_0}$ and we have:

\begin{theorem}
 For any $g \in \cS_{f_0}$, we have that $\cW_\kappa^n (g) \to f_0$ for the $L^1$-norm at a geometric rate.
\end{theorem}

This proves the Pareto stability conjecture, exactly as stated by Pareto (see the citation in the introduction): If we remove all wealth larger 
than some value $x$ from the invariant solution, then the perturbation thus obtained is $L^1$ bounded because of summability of the tail, hence the 
stability.

\section{Other more refined models.}

With the same ideas, we can build more sophisticated models that will be studied in the future. The main difference with the model 
presented here is that the invariant solutions cannot be computed explicitely in general, nor we can give close formulas for the Pareto 
exponents. But this does not prevent numerical studies of the invariant solutions.

We may more realistically assume that there are different sorts of individuals with different skills for finantial investment (different $p$'s), 
and different risk profiles (different $\gamma$'s). If we assume that each class of individuals 
is equally represented accross wealth classes (which is not true, the more skilled ones should be more numerous in the upper classes), then we end with a general 
wealth operator of the form
$$
\cW_\kappa (f) =\sum_i \frac{p_i (1+\gamma_i)}{\kappa} f(x/(1+\gamma_i)) + \frac{q_i (1+\gamma_i)}{\kappa} f(x(1+\gamma_i)) \ ,
$$ 
with 
$$
\sum_i p_i +\sum_i q_i =1 \ .
$$
The exponentials of the Pareto exponents appear then as roots of a Dirichlet polynomial. One can prove, using results 
from \cite{S} that the invariant solutions obey Pareto law.

A more realistic model consists in allowing the dissipative coefficient $\kappa$ to be non constant and make
it dependent on $x$. In principle, $x\mapsto \kappa (x)$ should be increasing. Then the search for invariant solutions leads to a functional equation with non-constant
coefficients whose possible explicit resolution depends on the form of the function $x\mapsto \kappa (x)$.

\noindent \textbf{Acknowledgements.} I thank my colleague Philippe Marchal for pointing out an error in the formula of the first version of this article.

\end{document}